\documentclass[8pt]{article}

\usepackage[english]{babel}
\usepackage[latin1]{inputenc}
\usepackage[T1]{fontenc}
\usepackage[usenames]{color}
\usepackage{amsmath}
\usepackage{amssymb}
\usepackage{hhline}
\usepackage{multirow}
\usepackage{colortbl}
\usepackage{enumerate}
\usepackage{array}
\usepackage{soul}
\usepackage{graphicx}
\usepackage{eurosym}
\usepackage{mathrsfs}
\usepackage{dsfont}
\usepackage{amsthm}
\usepackage{makecell}
\usepackage[titletoc]{appendix}
\usepackage[titles]{tocloft}
\usepackage{indentfirst}

\newtheorem{theorem}{Theorem}[section]
\newtheorem{proposition}{Proposition}[section]

\newtheorem{corollary}{Corollary}[section]
\newtheorem{lemma}{Lemma}[section]

\title{An extension of Heston's SV model to \\ Stochastic Interest Rates}

\author{Javier de Frutos, V\'ictor Gat\'on}

\begin{document}

\maketitle

\begin{abstract}In \cite{HestonSV}, Heston proposes a Stochastic Volatility (SV) model with constant interest rate and derives a semi-explicit valuation formula. Heston also describes, in general terms, how the model could be extended to incorporate Stochastic Interest Rates (SIR). This paper is devoted to the construction of an extension of Heston's SV model with a particular stochastic bond model which, just increasing in one the number of parameters, allows to incorporate SIR and to derive a semi-explicit formula for option pricing.

\

\end{abstract}

\section{Introduction}

In \cite{HestonSV}, Heston proposes a Stochastic Volatility (SV) model with constant interest rate and derives a semi-explicit valuation formula. Heston also describes, in general terms, how the model could be extended to incorporate Stochastic Interest Rates (SIR). We wil see how, with a particular stochastic bond model and just increasing in one the number of parameters, we can incorporate SIR and derive a semi-explicit formula for option pricing.

The paper will be organized as follows. First, we will review Heston's original model with constant interest rates. In a second step, we will make the theoretical development of the extended model as presented in \cite{HestonSV}. In a third step, we will search for a stochastic bond formula that can be nested within this framework, i.e., that fits with the specifications of the pricing model and does not increase much the number of parameters.

Finally, we will assume that the market is composed by the stock and the discounted bond computed in the previous step. We will see that, under certain parameter restrictions, the resulting model is of the type proposed by Heston in \cite{HestonSV}. We will derive a semi-explicit formula and obtain a pricing model which has just one more parameter than the original Heston's SV. Thus, we will have incorporated stochastic interest rates without increasing much the number of parameters.

\section{Heston SV model}

We recall that in Heston's model \cite{HestonSV}, the dynamics is:
\begin{equation}
   \left\{
    \begin{aligned}
         d\bar{S}(t) &=\mu \bar{S}(t)dt+\sqrt{\bar{v}(t)}\bar{S}(t)d\bar{z}_1(t), \\
         d\bar{v}(t) &=\kappa[\theta-\bar{v}(t)]dt+\sigma\sqrt{\bar{v}(t)}d\bar{z}_2(t), \\
    \end{aligned}
   \right.
\end{equation}
where $\bar{z}_1$ and $\bar{z}_2$ are Wiener processes.

Employing the notation of \cite{Bjork} or \cite{HestonSV}, we define the (instantaneous) correlation coefficient $\rho$ by $\rho dt=\text{Cov}(d\bar{z}_1,d\bar{z}_2)$, where $\text{Cov}(\ldotp,\ldotp)$ stands for covariance.

We also assume that a constant rate risk-free bond exists: $B(t,T)=e^{-r_0(T-t)}$.

In \cite{HestonSV}, it is claimed that these assumptions are insufficient to price contingent claims, because we have not made an assumption that gives the price of ``volatility risk''. By no arbitrage arguments (see \cite{Bjork} or \cite{HestonSV}), the value of any claim must satisfy:
\begin{equation}
\frac{1}{2}vS^2\frac{\partial^2 U}{\partial S^2}+\rho \sigma vS\frac{\partial^2 U}{\partial S \partial v} +\frac{1}{2}\sigma^2v\frac{\partial^2 U}{\partial v^2}+r_0S\frac{\partial U}{\partial S}+\left(\kappa(\theta-v)-\lambda(S,v,t)\right)\frac{\partial U}{\partial v}-r_0U+\frac{\partial U}{\partial t}=0
\end{equation}
where $\bar{S}(t)=S, \ \bar{v}(t)=v$ and $\lambda(S,v,t)$ represents the price of volatility risk.

We will assume that any risk premia is of the form $\lambda(S,v,t)=\lambda v$. It should be remarked that, once fixed the components of the market, the risk premia is independent of the claim, i. e. the same risk premia is used to price all the claims (see \cite{Bjork}).

As Heston points in \cite{HestonSV}, this choice of risk premia is not arbitrary (see \cite{Breeden} and \cite{CoxInglesonRoss2}).

Thus, the price of the European Call Option $U(S,v,t)$ satisfies the PDE:
\begin{equation}
\frac{1}{2}vS^2\frac{\partial^2 U}{\partial S^2}+\rho \sigma vS\frac{\partial^2 U}{\partial S \partial v} +\frac{1}{2}\sigma^2v\frac{\partial^2 U}{\partial v^2}+r_0S\frac{\partial U}{\partial S}+\left(\kappa(\theta-v)-\lambda v\right)\frac{\partial U}{\partial v}-r_0U+\frac{\partial U}{\partial t}=0,
\end{equation}
subject to the following conditions:
\begin{equation}\label{Ch2boundarydata}
\begin{aligned}
U(S,v,T)&=\max(0, S-K), \quad  \\
U(0,v,t)&=0, \quad && \left.r_0S\frac{\partial U}{\partial S}+\kappa\theta\frac{\partial U}{\partial v}-r_0U+U_t\right|_{(S,0,t)}= 0, \\
\frac{\partial U}{\partial S}(\infty,v,t)&=1, \quad &&U(S,\infty,t)= S.\\
\end{aligned}
\end{equation}

Heston conjectures a solution similar to the Black-Scholes model:
\begin{equation}\label{Ch2ecuSVsolution}
    U(S,v,t,T,K)=S\cdot R_1-K\cdot B(t,T) \cdot R_2,
\end{equation}

The following semi-explicit formula for the price of the European Option is obtained
\begin{equation}
    U\left(x,v,\tau,\ln (K)\right)=x\cdot R_1\left(x,v,\tau;\ln(K)\right)-\ln(K)\cdot B(t,T) \cdot R_2\left(x,v,\tau;\ln(K)\right),
\end{equation}
where $x=\ln (S)$, $\tau=T-t$ and function $R_j, \ j\in\{1,2\}$ is given by
\begin{equation}\label{Ch3rjquellamoluego}
    R_j(x,v,\tau;\ln(K))=\frac{1}{2}+\frac{1}{\pi}\int^{\infty}_0{Re\left[\frac{e^{-i\phi \ln(K)}f_j(x,v,\tau,\phi)}{i\phi}\right]d\phi},
\end{equation}
where
\begin{equation}\label{Ch3rjquellamoluego2}
\begin{aligned}
f_j(x,v,\tau;\phi) &= e^{C(\tau;\phi)+D(\tau;\phi)+i\phi x}, \\
C(\tau;\phi) &= r_0\phi i\tau + \frac{a}{\sigma^2}\left\{(b_j-\rho \sigma \phi i+d)\tau-2\ln\left[\frac{1-ge^{d\tau}}{1-g}\right]\right\}, \\
D(\tau;\phi) &= \frac{b_j-\rho\sigma\phi i+d}{\sigma^2}\left[\frac{1-e^{d\tau}}{1-ge^{d\tau}}\right], \\
g &= \frac{b_j-\rho\sigma\phi i+d}{b_j-\rho\sigma\phi i-d}, \quad d=\sqrt{(\rho\sigma\phi i-b_j)^2-\sigma^2(2\zeta_j\phi i- \phi^2)}, \\
\zeta_1&=\frac{1}{2}, \quad \zeta_2=-\frac{1}{2}, \quad a=\kappa\theta, \quad b_1=\kappa+\lambda-\rho\sigma, \quad b_2=\kappa+\lambda.
\end{aligned}
\end{equation}

\section{The extended model}\label{OPWVIRSVSIRtem}

We propose (see \cite{HestonSV}) the following market dynamics in the physical measure:
\begin{equation}\label{Ch2ecumarketdynamics}
   \left\{
    \begin{aligned}
         d\bar{S}(t) &=\mu_S \bar{S}(t)dt+\sigma_s(t)\sqrt{\bar{v}(t)}\bar{S}(t)d\bar{z}_1(t), \\
         d\bar{v}(t) &=\kappa[\theta-\bar{v}(t)]dt+\sigma\sqrt{\bar{v}(t)}d\bar{z}_2(t), \\
         d\bar{B}(t,T) &= \mu_b \bar{B}(t,T)dt+\sigma_b(t)\sqrt{\bar{v}(t)}\bar{B}(t,T) d\bar{z}_3(t),\\
    \end{aligned}
   \right.
\end{equation}

We also denote
\begin{equation}
  \rho_{sv}dt=\text{Cov}(d\bar{z}_1,d\bar{z}_2),\quad \rho_{sb}dt=\text{Cov}(d\bar{z}_1,d\bar{z}_3), \quad \rho_{vb}dt=\text{Cov}(d\bar{z}_2,d\bar{z}_3).
\end{equation}

Let $\bar{\bold{X}}(t)=(\bar{S}(t),\bar{v}(t),\bar{B}(t,T))$. Let us assume that the short rate of interest is a deterministic function of the state factors, i.e. $\bar{r}=\bar{r}(\bar{\bold{X}}(t))$, (short rates are stochastic but, at any fixed time $t$, they can be computed from the state of the market). Assuming as in \cite{HestonSV} that the risk premia is of the form $\lambda v$, any claim satisfies the PDE (see \cite{Bjork}, pg 218):
\begin{equation}\label{Ch2pdeextendedmodel}
\begin{aligned}
 \frac{\partial U}{\partial t} & +\frac{1}{2}\sigma^2_s vS^2\frac{\partial^2 U}{\partial S^2}+\frac{1}{2}\sigma^2v\frac{\partial^2 U}{\partial v^2}+\frac{1}{2}\sigma^2_b vB^2\frac{\partial^2 U}{\partial B^2}+ \rho_{sv}\sigma_s\sigma Sv \frac{\partial^2 U}{\partial s \partial v}+\rho_{sb}\sigma_s\sigma_b v SB\frac{\partial^2 U}{\partial S \partial B}\\
& +\rho_{vb}\sigma_b \sigma Bv \frac{\partial^2 U}{\partial v \partial B}+ rS\frac{\partial U}{\partial S}+[k(\theta-v)-\lambda v]\frac{\partial U}{\partial v}-rU+rB\frac{\partial U}{\partial B}=0,
\end{aligned}
\end{equation}
where $\bar{\bold{X}}(t)=\bold{X}=(S,v,B)$, $r=r\left(\bold{X}\right)$ and subject to the terminal condition of the claim (European Call), proper boundary data (see (\ref{Ch2boundarydata})) and $B(T,T)=1$.

There also exists a risk-neutral measure $\pi$. The value of any T-claim $U(t,\bold{X})$ is given by the conditional expectation:
\begin{equation}
U(t,\bold{X})=E^{\pi}\left[\left.e^{-\int^T_t\bar{r}\left(\bar{\bold{X}}(s)\right)ds}U(\bar{\bold{X}}(T))\right|\bar{\bold{X}}(t)=\bold{X}\right],
\end{equation}
and the market dynamics in the risk neutral measure is given by
\begin{equation}\label{Ch2extendedmodel}
   \left\{
    \begin{aligned}
         d\bar{S}(t) &=r\bar{S}(t)dt+\sigma_s(t)\sqrt{\bar{v}(t)}\bar{S}(t)d\bar{z}_1(t), \\
         d\bar{v}(t) &=[k\theta-k\bar{v}(t)-\lambda \bar{v}(t)]dt+\sigma\sqrt{\bar{v}(t)}d\bar{z}_2(t), \\
         d\bar{B}(t,T) &= r\bar{B}(t,T)dt+\sigma_b(t)\sqrt{\bar{v}(t)}\bar{B}(t,T) d\bar{z}_3(t).\\
    \end{aligned}
   \right.
\end{equation}

The change of variable $x=\ln\left(\frac{S}{B(t,T)}\right)$ implies that the PDE in the new variable is:

\begin{equation}\label{Ch2pdeextendedmodelcambiodevariable}
\begin{aligned}
& \frac{\partial U}{\partial t}+\left(\frac{1}{2}\sigma^2_sv+\frac{1}{2}\sigma^2_bv-\rho_{sb}\sigma_s\sigma_bv\right)\frac{\partial^2 U}{\partial x^2}+\frac{1}{2}\sigma^2v\frac{\partial^2 U}{\partial v^2}+\frac{1}{2}\sigma^2_bvB^2\frac{\partial^2 U}{\partial B^2} \\
& +
\left(-\sigma^2_bvP+\rho_{sb}\sigma_s\sigma_bvB\right)\frac{\partial^2 U}{\partial x \partial B}+\left(\rho_{sv}\sigma_s\sigma v-\rho_{vb}\sigma_b \sigma v\right)\frac{\partial^2 U}{\partial x \partial v}+\left(\rho_{vb}\sigma_b \sigma v B\right)\frac{\partial^2 U}{\partial v \partial B}\\
& + \left(-\frac{1}{2}\sigma^2_sv+\frac{1}{2}\sigma^2_bv\right)\frac{\partial U}{\partial x} + [k(\theta-v)-\lambda v]\frac{\partial U}{\partial v}+rB\frac{\partial U}{\partial B}-rU=0.
\end{aligned}
\end{equation}

Similar to the simple SV model, Heston conjectures a solution of the form:
\begin{equation}\label{Ch2conjsolhesstorat}
U(t,x,P,v)=e^xB(t,T)R_1(t,x,v)-KB(t,T)R_2(t,x,v),
\end{equation}

Substituting (\ref{Ch2conjsolhesstorat}) into equation (\ref{Ch2pdeextendedmodelcambiodevariable}), we obtain that $R_j(t,x,v)$ must satisfy, for $j=1,2$:
\begin{equation}\label{Ch2partdiffequsolhestonsvsir}
\frac{1}{2}\sigma^2_xv\frac{\partial^2 R_j}{\partial x^2}+\rho_{xv}\sigma_x\sigma v \frac{\partial^2 R_j}{\partial x \partial v}+\frac{1}{2} \sigma^2 v \frac{\partial^2 R_j}{\partial v^2}+\zeta_jv\frac{\partial R_j}{\partial x}+(a-b_jv)\frac{\partial R_j}{\partial v}+\frac{\partial R_j}{\partial t}=0,
\end{equation}
where
\begin{equation}\label{Ch2partdiffequsolhestonsvsir2}
\begin{aligned}
\frac{1}{2}\sigma^2_x &=\frac{1}{2}\sigma^2_s-\rho_{sb}\sigma_s\sigma_b+\frac{1}{2}\sigma^2_b, \quad \rho_{xv}=\frac{\rho_{sv}\sigma_s\sigma-\rho_{bv}\sigma_b\sigma}{\sigma_x\sigma}, \\
\zeta_1 &=\frac{1}{2}\sigma^2_x, \quad \zeta_2=-\frac{1}{2}\sigma^2_x, \quad a=k\theta ,\\
b_1 &=k+\lambda-\rho_{sv}\sigma_s\sigma, \quad b_2=k+\lambda-\rho_{bv}\sigma_b\sigma, \\
\end{aligned}
\end{equation}
subject to the condition at maturity corresponding to the European Option Call:
\begin{equation*}
R_j(T,x,v;\ln(K))=I_{\{x\geq \ln(K)\}},
\end{equation*}
where $I$ denotes the indicator function.

In Section  \ref{Ch3SVSRVF} we will see that, with the bond model that we are going to propose, short rates are of the form $r=\mu+\beta v$ ($\mu, \ \beta$ constant) and, using no arbitrage arguments, that the risk premia must be $\lambda(S,P,v,t)=\lambda v$, so we can apply Heston's results.

\section{The stochastic bond.}

We are looking for a bond formula which can be nested in (\ref{Ch2ecumarketdynamics}). Longstaff and Schwartz develop in \cite{LongstaffSchwartz} a model for interest rates that we are partly going to use.

Without loss of generality, we can assume that the bond is offered to the market by an entity (the US government for example), whose unique function in the market is to trade the bond. This bond is constructed, by no arbitrage arguments, upon a certain asset $\bar{Q}$ with dynamics:
\begin{equation}\label{Ch2ecuaccLS1}
\left\{
\begin{aligned}
    d\bar{Q} &=(\mu+\delta \bar{v})\bar{Q} dt+\sigma_{\bar{Q}} \sqrt{\bar{v}}\bar{Q} d\bar{Z}, \\
    d\bar{v} &=[k(\theta-\bar{v})]dt+\sigma\sqrt{\bar{v}}d\bar{z}_2.
\end{aligned}
\right.
\end{equation}
where $\bar{v}(t)$ is the same volatility process of (\ref{Ch2ecumarketdynamics}).

We assume that asset $\bar{Q}$, although dependant of the state of the market, is only accessible to the the entity which offers the bond. Therefore, any other investor who invests in the market described by (\ref{Ch2ecumarketdynamics}) can only negotiate upon the traded stock $\bar{S}$ and the bond.

Following the development in \cite{LongstaffSchwartz}, we assume that individuals have time-additive preferences of the form
\begin{equation}\label{Ch2ecuLSindpref}
E_t\left[\int^{\infty}_{t}\exp(-\rho s)\log(\bar{C}_s)ds\right],
\end{equation}
where $E[\cdotp]$ is the conditional expectation operator, $\rho$ is the utility discount factor and $\bar{C}_s$ represents consumption at time $s$.

The representative investor's decision problem is equivalent to maximizing (\ref{Ch2ecuLSindpref}) subject to the budget constraint
\begin{equation}
d\bar{W}=\bar{W}\frac{d\bar{Q}}{\bar{Q}}-\bar{C}dt,
\end{equation}
where $\bar{W}$ denotes wealth.

Standard maximization arguments employed in \cite{LongstaffSchwartz} lead to the following equation for the wealth dynamics
\begin{equation}
d\bar{W}=(\mu+ \delta \bar{v}(t) -\rho)\bar{W}dt+\sigma_{\bar{Q}} \bar{W} \sqrt{\bar{v}(t)}d\bar{Z}.
\end{equation}

Applying Theorem 3 in \cite{CoxInglesonRoss}, the value of a contingent claim $B(t,v)$ must satisfy the PDE
\begin{equation}\label{Ch2contingclaimecu}
-B_{t}= \frac{\sigma^2 v}{2}B_{vv}+(k\theta -kv-\lambda v)B_v-rB,
\end{equation}
where $\bar{v}(t)=v$, the market price of risk is $\lambda v$ and $\bar{r}(t)=r$ is the instantaneous riskless rate.

To obtain the equilibrium interest rate $\bar{r}$, Theorem 1 of \cite{CoxInglesonRoss} is applied. This theorem relates the riskless rate to the expected rate of change in marginal utility. The result obtained is that
\begin{equation}\label{Ch2ecurdelongenSVSIR}
\bar{r}(t) = \mu +(\delta-\sigma^2_{\bar{Q}})\bar{v}(t)=\mu+\beta \bar{v}(t), \\
\end{equation}

The price of a riskless unit discount bond $B(\tau,v)$, where $\tau=T-t$ is obtained solving equation (\ref{Ch2contingclaimecu}) subject to the maturity condition $B(0,v)=1$.

For the rest of the paper, we assume that $\beta>0$. We will see that when parameter $\beta \rightarrow 0^{+}$, the function $B(\tau,v)$ approaches to the bond price when the risk-free rate is considered constant ($B(\tau,v)=e^{-\mu\tau}$).

Now, we proceed to give the main result of this Section.

\begin{theorem}
The riskless unit discount bond $B(\tau,v)$, where $\tau=T-t$ denotes the time until maturity, $\bar{v}(\tau)=v$ and $\bar{r}(t)=r= \mu +\beta v$, is given by the formula:
\begin{equation}\label{Ch2formulabono}
B(\tau,v)=F(\tau)e^{G(\tau)v},
\end{equation}
where
\begin{equation}\label{Ch2formulabono2}
\begin{aligned}
    F(\tau)&=\exp\left(-\left(\mu+\frac{k\theta}{b}\right)\tau+k\theta\left(\frac{b+c}{bc}\right)\ln\left(b+ce^{d\tau}\right)-k\theta\left(\frac{b+c}{bc}\right)\ln(b+c)\right), \\
    G(\tau) &=\frac{e^{d\tau}-1}{b+ce^{d\tau}},
\end{aligned}
\end{equation}
and
\begin{equation}\label{Ch2formulabono3}
    d=-\sqrt{(k+\lambda)^2+2\beta\sigma^2}, \quad  b=\frac{(k+\lambda)-d}{2\beta}, \quad c=\frac{-(k+\lambda)-d}{2\beta}. \\
\end{equation}
\end{theorem}

\begin{proof}

For simplicity, along the proof, we will employ the notation:
\begin{equation*}
\eta=k\theta, \quad \alpha=k+\lambda. \ \
\end{equation*}

The claim satisfies the partial differential equation (\ref{Ch2contingclaimecu}) subject to the maturity condition $B(0,v)=1$. With the notation that we have just introduced, we have to solve:
\begin{equation*}
\left\{
\begin{aligned}
     & B_{\tau}=\frac{\sigma^2}{2}vB_{vv}+(\eta-\alpha v)B_v-(\mu+\beta v)B, \\
     & B(0,v)=1.
\end{aligned}
\right.
\end{equation*}

We conjecture a solution of the form $B(\tau,v)=F(\tau)e^{G(\tau)v},$ thus, $B_v,\ B_{vv}$ and $B_{\tau}$ are explicitly computable. Condition $B(0,v)=1$ imposes that $F(0)=1$ and $G(0)=0$.

Substituting into the PDE
\begin{equation}\label{Ch2ecudesarrolloprecio}
\frac{\sigma^2}{2}v F(\tau)G^2(\tau)+(\eta-\alpha v)F(\tau)G(\tau)-(\mu+\beta v)F(\tau)=F'(\tau)+F(\tau)G'(\tau)v.
\end{equation}

As the previous equation is an identity in $v$, we obtain two equations:
\begin{equation*}
\left\{
\begin{aligned}
    & \frac{\sigma^2}{2}F(\tau)G^2(\tau)-\alpha F(\tau)G(\tau)-\beta F(\tau)=F(\tau)G'(\tau), \\
    & \eta F(\tau)G(\tau)-\mu F(\tau)=F'(\tau). \\
\end{aligned}
\right.
\end{equation*}

For the first one, as candidate for solution we take:
\begin{equation*}
G(\tau)=\frac{a+e^{d\tau}}{b+c e^{d\tau}}=\frac{e^{d\tau}-1}{b+c e^{d\tau}},
\end{equation*}
as $G(0)=0$ implies $a=-1$ and $b\neq-c$.

Thus, obtaining $G^2(\tau)$, $G'(\tau)$  and substituting, we obtain a second degree equation given in function of $\exp(2d\tau),\exp(d\tau),1$, which implies that:
\begin{equation*}
\begin{aligned}
\sigma^2-2\alpha c-2\beta c^2 &= 0, \\
-2\sigma^2-2\alpha(b-c)-4\beta bc &= 2(bd+cd), \\
\sigma^2+2\alpha b-2\beta b^2 &=0.
\end{aligned}
\end{equation*}

Solved for $b$ and $c$, we obtain:
\begin{equation*}
c=\frac{-\alpha\pm\sqrt{\alpha^2+2\beta\sigma^2}}{2\beta}, \quad b=\frac{\alpha\pm\sqrt{\alpha^2+2\beta\sigma^2}}{2\beta}. \\
\end{equation*}

As $b\neq-c$, two solutions are eliminated. Another one is rejected when solving the other ODE as it appears $\ln(b+c)$, which must be positive. The solution is then:
\begin{equation*}
c=\frac{-\alpha+\sqrt{\alpha^2+2\beta\sigma^2}}{2\beta}, \quad b=\frac{\alpha+\sqrt{\alpha^2+2\beta\sigma^2}}{2\beta}, \quad d=-\sqrt{\alpha^2+2\beta \sigma^2}.
\end{equation*}

For the second equation, we obtain:
\begin{equation*}
\left\{
\begin{aligned}
    & \eta F(\tau)G(\tau)-\mu F(\tau)=F'(\tau), \\
    & F(0)=1.
\end{aligned}
\right.
\end{equation*}

After substituting, we arrive to:
\begin{equation*}
F(\tau)=\exp\left(-\left(\mu+\frac{\eta}{b}\right)\tau+\eta\frac{b+c}{bc}\ln(b+ce^{d\tau})-\eta\frac{b+c}{bc}\ln(b+c)\right),
\end{equation*}
which completes the proof.

\end{proof}

For the rest of the Chapter, we denote $\bar{B}(\tau,\bar{v})=B(\tau,\bar{v})$. To finish the Section, we give some auxiliary results which are quite straightforward to prove.

\begin{proposition}\label{Ch2dinamicbonoprop}
The bond dynamics in the physical measure is given by
\begin{equation}\label{Ch2dinamicbono}
\begin{aligned}
d\bar{B}(\tau,\bar{v}) &= \left[\mu+\beta \bar{v}+\lambda \bar{v}\right]\bar{B}(\tau,\bar{v}) dt + G(\tau)\sigma\sqrt{\bar{v}}\bar{B}(\tau,\bar{v})d\bar{z}_2 \\
&= (\bar{r}(t)+\lambda \bar{v})\bar{B}(\tau,\bar{v}) dt+ G(\tau)\sigma\sqrt{\bar{v}}\bar{B}(\tau,\bar{v})d\bar{z}_2,
\end{aligned}
\end{equation}
where $\bar{r}(t)$ denotes the instantaneous riskless rate and $\bar{z}_2$ is the same Wiener process as in equation (\ref{Ch2ecumarketdynamics}).
\end{proposition}

The following result will be interesting when we incorporate the bond to the pricing model of the option. It states that when parameter $\beta$ approaches to $0^{+}$, then function $B(\tau,v)$ converges to the price of the bond when constant risk-free rates are employed, i.e., the bond employed in the simple SV model.

\begin{proposition}\label{Ch2recuperacionbonodeter}
Consider the functions $F(\tau)$ and $G(\tau)$ given by (\ref{Ch2formulabono2})-(\ref{Ch2formulabono3}).

If $\beta\rightarrow 0^{+}$, then we have that $F(\tau)\rightarrow \exp(-\mu \tau)$ and $G(\tau)\rightarrow 0$.
\end{proposition}

\begin{lemma}\label{Ch2remarklemabono}
Let $G(\tau)$ be given by (\ref{Ch2formulabono2}). Then it holds:
\begin{equation*}
\left\{
\begin{aligned}
G(\tau) &=\frac{e^{d\tau}-1}{b+c e^{d\tau}} \underset{\tau\rightarrow0}\longrightarrow 0, \\
G(\tau) &\neq 0, \quad \tau>0,
\\ G(0) &=0.
\end{aligned}
\right.
\end{equation*}
\end{lemma}

As $\bar{B}(\tau,\bar{v})$ is the stochastic process of a bond price, the stochastic component $G(\tau)\sigma\sqrt{\bar{v}}$ of equation (\ref{Ch2dinamicbono}) must vanish at maturity so the bond reaches par at maturity with probability one. This is also satisfied due to the previous Lemma.

\section{Valuation Formula}\label{Ch3SVSRVF}

Suppose that the market is formed by a stock given by (physical measure)
\begin{equation*}
   \left\{
    \begin{aligned}
         d\bar{S}(t) &=\mu_S \bar{S}(t)dt+\sigma_s(t)\sqrt{\bar{v}(t)}\bar{S}(t)d\bar{z}_1(t), \\
         d\bar{v}(t) &=\kappa[\theta-\bar{v}(t)]dt+\sigma\sqrt{\bar{v}(t)}d\bar{z}_2(t), \\
    \end{aligned}
   \right.
\end{equation*}
and by a bond
\begin{equation*}
\bar{B}(t,T;\bar{v})=\bar{B}(\tau;\bar{v})=F(\tau)e^{G(\tau)\bar{v}},
\end{equation*}
where $\tau=T-t$ and $F(\tau), \ G(\tau)$ are explicitly given by formulas (\ref{Ch2formulabono2})-(\ref{Ch2formulabono3}).

If we compute the bond dynamics, Proposition \ref{Ch2dinamicbonoprop} enforces that, in order to be consistent with model (\ref{Ch2ecumarketdynamics}),
\begin{equation}\label{Ch2consistentconditions}
\left\{
\begin{aligned}
 \sigma_b(\tau) &= \sigma G(\tau), \\
 \rho_{bv} &=1, \\
 \rho_{bs} &= \rho_{vs}, \\
\end{aligned}
\right.
\end{equation}
and for simplicity reasons we have taken $\sigma_{S}(t)\equiv 1$.

The sign and magnitude of the correlation between the bond and the stock seems to be difficult to estimate from market data (see \cite{Johnson}). Condition $\rho_{bs}= \rho_{vs}$, although restrictive, does not violate market empirical observations in the sense of the sign (positive/negative).

\begin{proposition}
The short interest rate is given by $r=\mu+\beta v$ and the risk premia $\lambda(S,v,B,t)=\lambda v$ where $\lambda$ is the constant employed in the bond formula (\ref{Ch2formulabono}).
\end{proposition}

\begin{proof}

Let us assume that it exists a deterministic function $\bar{r}=\bar{r}\left(\bar{\bold{X}}(t)\right)$ where $\bar{\bold{X}}(t)=(\bar{S}(t),\bar{v}(t),\bar{B}(t,T))$ for the short interest rate. Using the results in \cite{Bjork}, pg 218, any contingent claim must satisfy
\begin{equation*}
\begin{aligned}
& \frac{\partial U}{\partial t}+\frac{1}{2}\sigma^2_s vS^2\frac{\partial^2 U}{\partial S^2}+\frac{1}{2}\sigma^2v\frac{\partial^2 U}{\partial v^2}+\frac{1}{2}\sigma^2_b vB^2\frac{\partial^2 U}{\partial B^2}+ \rho_{sv}\sigma_s\sigma Sv \frac{\partial^2 U}{\partial s \partial v}+\rho_{sb}\sigma_s\sigma_b v SB\frac{\partial^2 U}{\partial S \partial B}+\\
& +\rho_{vb}\sigma_b \sigma Bv \frac{\partial^2 U}{\partial v \partial B}+ rS\frac{\partial U}{\partial S}+[k(\theta-v)-\lambda(S,v,B,t)]\frac{\partial U}{\partial v}-rU+rB\frac{\partial U}{\partial B}=0.
\end{aligned}
\end{equation*}
where $\bar{\bold{X}}(t)=\bold{X}=(S,v,B)$.

Suppose that, fixed a maturity $T$ ($\tau=T-t$), we want to price the contingent claim which values 1 at maturity. In order to avoid any arbitrage opportunity, this claim has to be the bond,
\begin{equation*}
U(S,v,B,\tau)=F(\tau)e^{G(\tau)v},
\end{equation*}
thus, it must hold that
\begin{equation*}
\begin{aligned}
& -\left(F'(\tau)e^{G(\tau)v}+F(\tau)G'(\tau)v e^{G(\tau)v}\right)+\frac{1}{2}\sigma^2vF(\tau)G^2(\tau)e^{G(\tau)v} +\\
&+[k(\theta-v)-\lambda(S,v,B,t)]F(\tau)G(\tau)e^{G(\tau)v}-rF(\tau)e^{G(\tau)v}=0.
\end{aligned}
\end{equation*}

On the other hand, by construction of the bond, we know that
\begin{equation*}
\begin{aligned}
&-\left(F'(\tau)e^{G(\tau)v}+F(\tau)G'(\tau)v e^{G(\tau)v}\right)+\frac{1}{2}\sigma^2vF(\tau)G^2(\tau)e^{G(\tau)v}+ \\
&+[k(\theta-v)-\lambda v]F(\tau)G(\tau)e^{G(\tau)v}-(\mu+\beta v)F(\tau)e^{G(\tau)v}=0.
\end{aligned}
\end{equation*}

We subtract both expressions and divide by  $F(\tau)e^{G(\tau)v}$  to get to
\begin{equation*}
\left(-\lambda(S,v,B,t)+\lambda v\right)G(\tau)+\left(-r+(\mu+\beta v)\right)=0.
\end{equation*}

The previous expression must hold for all $v,\tau$. From Proposition \ref{Ch2remarklemabono} we know that $G(\tau){\neq} 0, \ \tau\neq 0$ and that $G(0)=0$. Standard arguments yield the desired result.

\end{proof}

In the riskless measure, the dynamics is:
\begin{equation}\label{Ch2riskfreedefparticular}
   \left\{
    \begin{aligned}
         d\bar{S}(t) &=r \bar{S}(t)dt+\sqrt{\bar{v}(t)}\bar{S}(t)d\bar{z}_1(t), \\
         d\bar{v}(t) &=[k\theta-k\bar{v}(t)-\lambda \bar{v}(t)]dt+\sigma\sqrt{\bar{v}(t)}d\bar{z}_2(t), \\
         d\bar{B}(t,T) &= r \bar{B}(t,T)dt+\sigma G(\tau)\sqrt{\bar{v}(t)}\bar{B}(t,T) d\bar{z}_2(t), \\
    \end{aligned}
   \right.
\end{equation}
where the riskless rate is $\bar{r}(t)=\mu+\beta \bar{v}(t)$.

If we compare it with the original SV  model of Heston, note that just one new parameter has appeared, $\beta$, which models the stochastic component of the bond.

Proposition \ref{Ch2recuperacionbonodeter} states that, as $\beta$ approaches to $0^{+}$, the function which gives the bond price $B(\tau,v)$ converges, for any fixed $v$, to $e^{-\mu\tau}$, which is the price of a bond when constant risk free rates are employed. Therefore, the original SV model can be considered a particular case of this one and we allow $\beta\geq 0$ where $\beta=0$ denotes the the original SV model.

Now we are going to develop a semi-explicit formula. We point that Heston conjectured in \cite{HestonSV} a solution for the extended model:
\begin{equation*}
U(t,x,P,v)=e^xB(t,T)R_1(t,x,v)-KB(t,T)R_2(t,x,v),
\end{equation*}
where $R_j, \ j\in\{1,2\}$ satisfies (\ref{Ch2partdiffequsolhestonsvsir})-(\ref{Ch2partdiffequsolhestonsvsir2}).

Substituting the parameter restrictions (\ref{Ch2consistentconditions}) into (\ref{Ch2partdiffequsolhestonsvsir})-(\ref{Ch2partdiffequsolhestonsvsir2}), we obtain
\begin{equation}
\frac{1}{2}\sigma^2_xv\frac{\partial^2 R_j}{\partial x^2}+\rho_{xv}\sigma_x\sigma v \frac{\partial^2 R_j}{\partial x \partial v}+\frac{1}{2} \sigma^2 v \frac{\partial^2 R_j}{\partial v^2}+\zeta_jv\frac{\partial R_j}{\partial x}+(a-b_jv)\frac{\partial R_j}{\partial v}+\frac{\partial R_j}{\partial t}=0,
\end{equation}
where
\begin{equation}
\begin{aligned}
\frac{1}{2}\sigma^2_x &=\frac{1}{2}-\rho_{sv}\sigma G(\tau)+\frac{1}{2}\sigma^2 G^2(\tau), \quad \rho_{xv}=\frac{\rho_{sv}-\sigma G(\tau)}{\sigma_x}, \\
\zeta_1 &=\frac{1}{2}\sigma^2_x, \quad \zeta_2=-\frac{1}{2}\sigma^2_x, \quad a=k\theta, \\
b_1 &=k+\lambda-\rho_{sv}\sigma, \quad b_2=k+\lambda-\sigma^2 G(\tau).
\end{aligned}
\end{equation}

The following result is proved in Appendix in \cite{HestonSV}.

\begin{lemma}

Let $\tau=T-t$. The solution of equation
\begin{equation}\label{Ch2edpprecioopsvsir}
\frac{1}{2}\sigma^2_xv\frac{\partial^2 f_j}{\partial x^2}+\rho_{xv}\sigma_x\sigma v \frac{\partial^2 f_j}{\partial x \partial v}+\frac{1}{2} \sigma^2 v \frac{\partial^2 f_j}{\partial v^2}+\zeta_jv\frac{\partial f_j}{\partial x}+(a-b_jv)\frac{\partial f_j}{\partial v}-\frac{\partial f_j}{\partial \tau}=0,
\end{equation}
subject to $f_j(x,v,0;\phi)=e^{i\phi x}, \ j\in\{1,2\}$ is the characteristic function of $R_j$.
\end{lemma}

In order to obtain the solution (\ref{Ch2edpprecioopsvsir}), the characteristic function is conjectured to be
\begin{equation*}
f_j(x,\upsilon,\tau,\phi)=e^{C_j(\tau,\phi)+D_j(\tau,\phi)\upsilon+i\phi x}.
\end{equation*}

Thus it holds that:
\begin{equation*}
\begin{aligned}
    & \frac{\partial f}{\partial t}=f\left(\frac{\partial C}{\partial t}+\frac{\partial D}{\partial t}\upsilon\right)=f\left(-\frac{\partial C}{\partial \tau}-\frac{\partial D}{\partial \tau}\upsilon\right), \\
    & \frac{\partial f}{\partial x}=fi\phi,  \quad \frac{\partial f}{\partial v}=fD, \\
    & \frac{\partial^2 f}{\partial x^2}=-f\phi^2, \quad  \frac{\partial^2 f}{\partial v^2}=fD^2, \quad \frac{\partial^2 f}{\partial v \partial x}=i\phi Df.
\end{aligned}
\end{equation*}

Substituting in the PDE, we come to:
\begin{equation*}
-\frac{1}{2}\sigma^2_xvf\phi^2+\rho_{xv}\sigma_x\sigma v i\phi Df+\frac{1}{2} \sigma^2 v fD^2+u_jvfi\phi+(a-b_jv)fD+f\left(-\frac{\partial C}{\partial \tau}-\frac{\partial D}{\partial \tau}\upsilon\right)=0.
\end{equation*}

As the previous expression in an identity in $v$ we obtain the next two equations:
\begin{equation*}
\left\{
\begin{aligned}
    & -\frac{1}{2}\sigma^2_x\phi^2+\rho_{xv}\sigma_x\sigma  i\phi D+\frac{1}{2} \sigma^2 D^2+u_ji\phi-b_jD-\frac{\partial D}{\partial \tau}=0, \\
    & aD-\frac{\partial C}{\partial \tau}=0, \\
\end{aligned}
\right.
\end{equation*}
plus the condition $C(0)=D(0)=0$.

The first equation is a Ricatti equation, but as $\sigma_x(t)$ depends on time and not being constant, a direct solution has not been found and it has to be solved numerically, for example, by means of the routine of matlab ode 45.

\begin{corollary}

The price of the option is then given by:
\begin{equation*}
    R_j(x,v,\tau,\ln(K))=\frac{1}{2}+\frac{1}{\pi}\int^{\infty}_0{Re\left[\frac{e^{-i\phi \ln(K)}f_j(x,v,\tau,\phi)}{i\phi}\right]d\phi},
\end{equation*}
where $f_j(x,\upsilon,\tau,\phi)=e^{C_j(\tau,\phi)+D_j(\tau,\phi)\upsilon+i\phi x}$.
\end{corollary}

\section*{Acknowledgements}

This work has been supported under grants MTM2016-78995-P (AEI/MINECO, ES), and VA105G18, VA024P17 (Junta de Castilla y Le\'{o}n, ES) cofinanced by FEDER funds.

\end{document}